\documentclass[aip,jmp, amsmath,amssymb,reprint]{revtex4-1} 
\usepackage{hyperref} 
\usepackage[utf8]{inputenc}
\usepackage{enumerate}
\usepackage{amsmath}
\usepackage{amssymb}
\usepackage{ stmaryrd }
\usepackage{amssymb}
\usepackage{amsthm}
\usepackage{graphicx}
\usepackage{dcolumn}
\usepackage{bm}
\usepackage{mathtools}
\usepackage{enumerate}

\begin{document}
\theoremstyle{plain} 
\newtheorem{thm}{Theorem}[section]
\newtheorem{prop}[thm]{Proposition}
\newtheorem{lem}[thm]{Lemma}
\newtheorem{cor}[thm]{Corollary}

\theoremstyle{definition}
\newtheorem{defn}[thm]{Definition} 

\newtheorem{example}[thm]{Example}
\newtheorem{query}[thm]{Question}
\theoremstyle{remark}
\newtheorem{remark}[thm]{Remark}

\numberwithin{equation}{section}

%
%

\renewcommand{\epsilon}{\varepsilon}
\renewcommand{\phi}{\varphi}
\renewcommand{\setminus}{\smallsetminus}
\renewcommand{\theequation}{ \thesection.\arabic{equation}}
\renewcommand{\thesection}{\arabic{section}}

\preprint{AIP/123-QED}

\title[Quantum algorithmic randomness]{Quantum algorithmic randomness}

\author{Tejas Bhojraj}

\affiliation{ 
Department of Mathematics, University of Wisconsin-Madison, USA.
}  
 \email{bhojraj@math.wisc.edu}

\date{1 February 2020}
\revised{\today}

\begin{abstract}
Quantum Martin-L{\"o}f randomness (q-MLR) for infinite qubit sequences was introduced by Nies and Scholz. We define a notion of quantum Solovay randomness which is equivalent to q-MLR. The proof of this goes through a purely linear algebraic result about approximating density matrices by subspaces. We then show that random states form a convex set. Martin-L{\"o}f absolute continuity is shown to be a special case of q-MLR. Quantum Schnorr randomness is introduced. A quantum analogue of the law of large numbers is shown to hold for quantum Schnorr random states.
\end{abstract}

\maketitle

\section{Introduction}
\subsection{Martin-L{\"o}f  randomness in classical and quantum settings}
Information theory has been generalized to the quantum realm \cite{Nielsen:2011:QCQ:1972505}. Similarly, the theory of computation has been extended to the quantum setting; a notable example being the conception of a quantum Turing machine \cite{Mller2007QuantumKC,doi:10.1137/S0097539796300921}. It hence seems natural to extend Kolmogorov complexity theory and algorithmic randomness, areas using notions from computation and information, to the quantum realm. While classical Kolmogorov complexity has inspired many competing definitions of quantum Kolmogorov complexity \cite{Berthiaume:2001:QKC:2942985.2943376,Mller2007QuantumKC,Vitnyi2001QuantumKC}, algorithmic randomness has only recently been extended to the quantum setting \cite{unpublished}.

What does algorithmic randomness study? Consider infinite sequences of ones and ze
roes (called  bitstrings in this paper). First consider the bitstring $101010101010\cdots$. It has an easily describable `pattern' to it; namely that the ones and zeroes alternate. Now take a bitstring obtained by tossing a fair coin repeatedly. Intuitively, it seems that the second bitstring, in contrast to the first, is unlikely to have patterns. A central theme in algorithmic randomness is to quantify our intuition that the second bitstring is more `random', more `structureless' than the first. While algorithmic randomness is concerned with the randomness of bitstrings, the present paper is concerned with randomness of \emph{qubitstrings} (infinite sequences of qubits); a line of inquiry initiated by Nies and Scholz \cite{unpublished}. The basic definitions from algorithmic randomness we state below may be found in books by Nies\cite{misc1} and Downey and Hirschfeldt\cite{misc}. Roughly speaking, a Martin-L{\"o}f random bitstring is one which has no algorithmically describable regularities. Slightly more rigorously, an infinite bitstring is said to be Martin-L{\"o}f random if it is not in any `effectively null' set. In the context of Martin-L{\"o}f randomness, a measurable set $A$ is effectively null if there is a Turing machine which computes a sequence of open sets, $(U_n)_n$ such that the measure of $U_n$ is at most $2^{-n}$ and $A \subseteq U_n$ for all $n$. By varying the definition of `effectively null', we get other notions of randomness like Solovay randomness and Schnorr randomness. Note that the randomness of a bitstring crucially depends on the notion of computability. In a broad sense, a bitstring is random if it has no `\emph{computably describable}' patterns. Consequently, the following notion is pertinent to us; a function, $f$ on the natural numbers is said to be \emph{computable} if there is a Turing machine, $\phi$ such that on input $n$, $\phi$ halts and outputs $f(n)$.
\begin{defn}
A sequence $(a_n)_{n\in \mathbb{N}}$ is said to be \emph{computable} if there is a computable function $f$, such that $f(n)=a_n$.

\end{defn}
The notion of a computable real number will come up when we discuss quantum Schnorr randomness. 
\begin{defn}
A real number $r$ is said to be \emph{computable} if there is a computable function $f$ such that for all $n$, $|f(n)-r|<2^{-n}$.
\end{defn}
We describe how classical algorithmic randomness generalizes to qubitstrings.
We refer the reader to the book by Nielsen and Chuang\cite{Nielsen:2011:QCQ:1972505} for preliminaries on quantum theory.
While it is clear what one means by a infinite sequence of bits, it is not immediately obvious how one would formalize the notion of an infinite sequence of qubits. To describe this, many authors have independently come up with the notion of a \emph{state} \cite{unpublished,article,brudno}. We will need the one given by Nies and Scholz \cite{unpublished}. A positive semidefinite matrix with trace equal to one is a density matrix and is commonly used to represent a probabilistic mixture of pure quantum states.
\begin{defn}
A \emph{state}, $\rho=(\rho_n)_{n\in \mathbb{N}}$ is an 
infinite sequence of density matrices such that $\rho_{n} \in \mathbb{C}^{2^{n} \times 2^{n}}$ and $\forall n$,  $PT_{\mathbb{C}^{2}}(\rho_n)=\rho_{n-1}$.
\end{defn}
The idea is that $\rho$ represents an infinite sequence of qubits whose first $n$ qubits are given by $\rho_n$. Here, $PT_{\mathbb{C}^{2}}$ denotes the partial trace which `traces out' the last qubit from $\mathbb{C}^{2^n}$. The definition requires $\rho$ to be coherent in the sense that for all $n$, $\rho_n$, when `restricted' via the partial trace to its first $n-1$ qubits, has the same measurement statistics as the state on $n-1$ qubits given by $\rho_{n-1}$.
The following state will be the quantum analogue of Lebesgue measure as explained later in remark \ref{rem:tr}.
\begin{defn}\cite{unpublished}
\label{def:tr}
Let $\tau=(\tau_n)_{n\in \mathbb{N}}$ be the state given by setting $\tau_n = \otimes_{i=1}^n I$ where $I$ is the two by two identity matrix.
\end{defn}
\begin{defn}\cite{unpublished} 
A \emph{special projection} is a hermitian projection matrix with complex algebraic entries.
\end{defn}
Since the complex algebraic numbers (roots of polynomials with rational coefficients) have a computable presentation, we may identify a special projection with a natural number and hence talk about computable sequences of special projections.
\begin{defn}\cite{unpublished}
\label{defn:sigclass}
A quantum $\Sigma_{1}^0$
set (or q-$\Sigma_{1}^0$
set for short) G is a computable
sequence of special projections $G=(p_{i})_{i\in \mathbb{N}}$ such that $p_i$ is $2^i$ by $2^i$ and range$(p_i \otimes I) \subseteq$ range $(p_{i+1})$ for all $i\in \mathbb{N}$. 

\end{defn}
\begin{defn}\cite{unpublished} 
If $\rho$ is a state and $G=(p_{n})_{n\in \mathbb{N}}$ a q-$\Sigma_{1}^0$
set as above, then $\rho(G):=\lim_{n}$ Tr$(\rho_{n}p_n)$.
\end{defn}

While a $2^n$ by $2^n$ special projection may be thought of as a computable projective measurement on a system of $n$ qubits, a q-$\Sigma_{1}^0$
class corresponds to a computable sequence of projective measurements on longer and longer systems of qubits. We motivate the definition of a quantum $\Sigma_{1}^0$
set by relating it to the classical $\Sigma_{1}^0$
class. Let $2^{\omega}$, called Cantor space, denote the collection of infinite bitstrings, let $2^n$ denote the set of bit strings of length $n$, $2^{<\omega}= \bigcup_{n} 2^n$, and let $2^{\leq\omega}:=  2^{<\omega} \cup 2^{ \omega}$. Cantor space can be topologized by declaring the cylinders to be the basic open sets. If $\pi \in 2^n$ for some $n$, then the cylinder generated by $\pi$, denoted $\llbracket \pi \rrbracket$, is the set of all sequences extending $\pi$:\[\llbracket \pi \rrbracket=\{X \in 2^{\omega}: X\upharpoonright n = \pi\}.\]
If $C \subseteq 2^n$, let \[\llbracket C \rrbracket := \bigcup_{\pi \in C} \llbracket \pi \rrbracket,\] be the set of all $X \in 2^{\omega}$ such that the initial segment of $X$ of length $n$ is in $C$. One of the many equivalent ways of defining a $\Sigma_{1}^0$
class is as follows. 
\begin{defn}
\label{def:10}
A $\Sigma_{1}^0$ class $S \subseteq 2^{\omega}$ is any set of the form, \[S=\bigcup_{i\in \mathbb{N}} \llbracket A_{i} \rrbracket\] where
    \begin{enumerate}
        \item $A_{i} \subseteq 2^i, \forall i\in \mathbb{N}$ 
        \item
        The indices of $A_{i}$ form a computable sequence. (Being a finite set, each $A_i$ has a natural number coding it.)
        \item $\llbracket A_{i} \rrbracket\subseteq \llbracket A_{i+1} \rrbracket, \forall i\in \mathbb{N}$.
    \end{enumerate}
\end{defn}
Letting $\llbracket A_{i} \rrbracket:= S_i$, we write $S=(S_i )_i$. A $\Sigma_{1}^0$ class, S is coded (non-uniquely) by the index of the total computable function generating the sequence $(A_{i})_{ i\in \mathbb{N}}$ occurring in (2) in the definition of $S$. Hence, the notion of a computable sequence of $\Sigma_{1}^0$ classes makes sense.
\begin{remark}\label{rem:tr}
The special projections $p_i$s in  Definition \ref{defn:sigclass} play the role of the $A_i$s which generate a the $\Sigma_{1}^0$ class, $S$ and that rank$(p_i)$ plays the role of $|A_i|$. Note that for $G=(p_{n})_{n\in \mathbb{N}}$, \emph{$\tau(G)=\lim_{n}(2^{-n}|p_n|)$} where, $|p_n|$ is the rank of $p_n$. So, $\tau(G)$ is the quantum analog of the Lebesgue measure of $S$ which equals $\lim_{n}(2^{-n}|A_n|)$.
\end{remark} 
Informally, and somewhat inaccurately, a q-$\Sigma_{1}^0$ class,  $G=(p_{n})_{n\in \mathbb{N}}$, may be thought of as a projective measurement whose expected value, when `measured' on a state $\rho=(\rho_{n})_{n\in \mathbb{N}}$ is $\rho(G)=\lim_{n}$ Trace $(\rho_{n}p_n)$. In reality, a q-$\Sigma_{1}^0$ class, $G=(p_{n})_n$, is a sequence of projective measurements on larger and larger finite dimensional complex Hilbert spaces. This sequence can be used to `measure' a coherent sequence of density matrices (i.e., a state) the expected value of which is the limit of the Trace$(\rho_{n}p_n)$ (the expected value of measuring the $n^{th}$ `level').

\begin{defn}
A \emph{classical Martin-L{\"o}f test} (MLT) is a computable sequence, $(S_{m})_{m \in \mathbb{N} }$ of $\Sigma_{1}^0$ classes such that the Lebesgue measure of $S_m$ is less than or equal to $2^{-m}$ for all m.
\end{defn}
By remark \ref{rem:tr}, the state $\tau$ is, roughly speaking, the quantum version of the Lebesgue measure. So, the quantum generalization of the classical MLT is:
\begin{defn}\cite{unpublished}
A \emph{quantum Martin-L{\"o}f test} (q-MLT) is a computable sequence, $(S_{m})_{m \in \mathbb{N}}$ of q-$\Sigma_{1}^0$ classes such that $\tau(S_m)$ is less than or equal to $2^{-m}$ for all m, where $\tau$ is as in Definition \ref{def:tr}.
\end{defn}

\begin{defn}\cite{unpublished}
$\rho$ is \emph{q-MLR} if for any q-MLT $(S_{m})_{m \in \mathbb{N}}$, $\inf_{m \in \mathbb{N}}\rho(S_m)=0$.
\end{defn}
Roughly speaking, a state is q-MLR if it cannot be `detected by projective measurements of arbitrarily small rank'.
\begin{defn}\cite{unpublished}
$\rho$ is said to \emph{fail  the q-MLT $(S_{m})_{m \in \mathbb{N}}$, at order $\delta$}, if $\inf_{m \in \mathbb{N}}\rho(S_m)>\delta$. $\rho$ is said to \emph{pass  the q-MLT $(S_{m})_{m \in \mathbb{N}}$ at order $\delta$} if it does not fail it at $\delta$. 
\end{defn}
So, $\rho$ is q-MLR if it passes all q-MLTs at all $\delta>0$.
\begin{remark}
\label{rem:notation}
A few remarks on notation: By `bitstring', we mean a finite or infinite classical sequence of ones and zeroes. It will be clear from context whether the specific bitstring under discussion is finite or infinite. $2^n$ will denote the set of bitstrings of length $n$. Let $B^n$ denote the standard computational basis for $\mathbb{C}^{2^n}$. I.e., $B^n:= \{|\sigma\big> : \sigma \in 2^n\} $. If $S\subseteq 2^n$, let $P_{S} :=\sum_{\sigma \in S} |\sigma\big>\big<\sigma|$. `Tr' stands for trace. A sequence of q-$\Sigma_{1}^0$ classes will be indexed by the superscript. The subscript will index the sequence of special projections comprising a q-$\Sigma_{1}^0$. For example, $(S^{m})_{m \in \mathbb{N}}$ is a sequence of q-$\Sigma_{1}^0$ classes and  $S^{m}=(S^{m}_{n})_{m \in \mathbb{N}}$ is a class from the sequence. So, a sequence of q-$\Sigma_{1}^0$ classes can be thought of as a double sequence of special projections: $(S^{m}_{n})_{m,n \in \mathbb{N}}$. Lebesgue measure is denoted by $\mu$.
\end{remark}
\subsection{Overview}
This paper has two major themes. First, it continues the study of quantum Martin-L{\"o}f randomness initiated by Nies and Scholz \cite{unpublished}. Second, we define quantum Solovay and quantum Schnorr randomness and prove results concerning these notions. Along with  Martin-L{\"o}f randomness, Solovay randomness and Schnorr randomness are important classical randomness notions. While Solovay randomness is equivalent to MLR, Schnorr randomness is strictly weaker. In Section \ref{sec:SolSchn}  we define quantum Solovay and quantum Schnorr randomness, show that quantum Solovay randomness is equivalent to q-MLR, show the convexity of the randomness classes in the space of states (answering open questions\cite{unpublished1, unpublished}), and obtain results regarding q-MLR states. The equivalence of quantum Solovay and quantum Martin-L{\"o}f randomness turns out to be a corollary of Theorem \ref{thm:LinAStr}, a linear algebraic result of independent interest concerning the approximation of density matrices by subspaces. This result, to the best of our knowledge, is novel and may prove useful in areas where approximations to density matrices are used; for example, quantum information and error correction,  quantum Kolmogorov complexity \cite{Mller2007QuantumKC,Berthiaume:2001:QKC:2942985.2943376} and quantum statistical mechanics.

In Section \ref{sec:measures}, we study states which are coherent sequences of diagonal density matrices. These states can be thought of as probability measures on Cantor space. Nies and Stephan\cite{unpublished2} defined  Martin-L{\"o}f absolutely continuity and Solovay randomness for diagonal states. We show that these two notions are the restrictions of q-MLR and quantum Solovay randomness to the space of diagonal states. We prove a result (Lemma \ref{lem:30}) about approximating a subspace of small rank by another one with a different orthonormal spanning set and of appropriately small rank. This result, novel as far as we know, may be applied to the important problem of approximating an entangled subspace (a subspace spanned by entangled pure states) by one spanned by product tensors \cite{demianowicz2019entanglement,book2}. We then discuss the quantum randomness of classical states in subsection \ref{schnorr}. Quantum Schnorr randomness would not be a natural quantum randomness notion were it not for Lemma \ref{thm:schcla} which shows that quantum Schnorr randomness agrees with classical Schnorr randomness on the classical states (i.e., states induced by infinite bitstrings).

Nies and Tomamichel \cite{logicblog} showed that q-MLR states satisfy a quantum version of the law of large numbers. In Section  \ref{sec:lln} we strengthen this by showing that in fact, all quantum Schnorr random states (a set strictly containing the q-MLR states) satisfy the law of large numbers.
\section{Notions of Quantum algorithmic randomness}
\label{sec:SolSchn}

\subsection{Solovay and Schnorr randomness}
An infinite bitstring $X$ is said to \emph{pass} the Martin-L{\"o}f test $(U_{n})_n$ if $X \notin \bigcap_{n}U_n$ and is said to be \emph{Martin-L{\"o}f random (MLR)} if it passes all Martin-L{\"o}f tests. A related randomness notion is Solovay randomness. A computable sequence of $\Sigma^{0}_1$ classes, $(S_{n})_n$ is a \emph{Solovay test} if $\sum_{n} \mu(S_{n})$, the sum of the Lebesgue measures is finite. An infinite bitstring $X$ \emph{passes} $(S_{n})_n$ if $X\in S_n$ for infinitely many $n$. It is a remarkable fact that $X$ is MLR if and only if it passes all Solovay tests. Is this also true in the quantum realm? Nies and Scholz asked \cite{unpublished1} if there is a notion of a quantum Solovay test and if so, is quantum Martin-L{\"o}f randomness equivalent to passing all quantum Solovay tests. We answer this question in the affirmative by defining a quantum Solovay test and quantum Solovay randomness as follows.
Roughly speaking, we obtain a notion of  a quantum Solovay test by replacing `$\Sigma^{0}_1$ class' and `Lebesgue measure' in the definition of classical Solovay tests with `quantum-$\Sigma^{0}_1$ set' and $\tau$ respectively. We show below that quantum Solovay Randomness is equivalent to q-MLR.

\begin{defn}\label{defn:1} 
A uniformly computable sequence of quantum-$\Sigma^{0}_1$ sets, $(S^{k})_{k\in\omega}$ is a \emph{ quantum-Solovay test} if  $\sum_{k\in \omega} \tau(S^{k}) <\infty.$\end{defn}\begin{defn}\label{defn:2}
For $0<\delta<1$, a state $\rho$ \emph{fails the Solovay test $(S^k)_{k\in\omega}$ at level $\delta$} if there are infinitely many $k$ such that $\rho(S^k)>\delta$.
\end{defn}
\begin{defn}
A state $\rho$ \emph{passes the Solovay test $(S^k)_{k\in\omega}$} if for all $\delta>0$, $\rho$ does not fail $(S^k)_{k\in\omega}$ at level $\delta$. I.e, lim$_{k}\rho(S^{k})=0$.
\end{defn}
\begin{defn}
A state $\rho$ is \emph{quantum Solovay random} if it passes all quantum Solovay tests.\end{defn}
An \emph{interval Solovay test}\cite{misc1} is a Solovay test, $(S_{n})_n$ such that each $S_n$ is generated by a finite collection of strings. By 7.2.22 in the book by Downey and Hirschfeldt \cite{misc}, a Schnorr test may be defined as:
\begin{defn}
A \emph{Schnorr test} is an interval Solovay test, $(S^{m})_{m}$ such that $\sum_{m}\mu(S^{m}) $ is a computable real number. 
\end{defn}
A bitstring passes a Schnorr test if it does not fail it (using the same notion of failing as in the Solovay test). We mimic this notion in the quantum setting.
\begin{defn}
\label{defn:Schnorr}
A \emph{quantum Schnorr test} is a strong Solovay test, $(S^{m})_{m}$ such that $\sum_{m}\tau(S^{m}) $ is a computable real number. A state is quantum Schnorr random if it passes all Schnorr tests.
\end{defn}
The following two definitions are due to Nies and Scholz\cite{unpublished}. The first is a quantum analogue of an interval Solovay test.
\begin{defn}
A \emph{strong Solovay test} is a computable sequence of special projections $(S^{m})_{m}$ such that $\sum_{m}\tau(S^{m}) < \infty$. A state $\rho$ fails $(S^{m})_{m}$ at $\epsilon$ if for infinitely many $m$, $\rho(S^{m})>\epsilon$.
\end{defn}
\begin{defn}
A state $\rho$ is \emph{weak Solovay random} if it passes all strong quantum Solovay tests.\end{defn}

\subsection{A general result about density matrices}
We prove a purely linear algebraic theorem about approximating density matrices by subspaces and then use it to show the equivalence of quantum Solovay and quantum Martin-L{\"
o}f randomness in the next subsection.

In words, the theorem says the following. Let $\mathcal{F}$ be a set of subspaces of `small' (at most $d$) total dimension and let $Q$ be the set of those density matrices `$\delta$ close' to at least $m$ many elements of $\mathcal{F}$. Then, there is a subspace of small (at most $4d/\delta m$) dimension `$\delta/4$ close' to every density matrix in $Q$. This subspace is spanned by a maximal (non-extendable) orthonormal subset of the set of vectors which are `close to $\mathcal{F}$' in a certain sense.

\begin{thm}
\label{thm:LinAStr}
Let $m,d,n \in \mathbb{N}$ and  $ \delta\in (0,1) $ be arbitrary. Let $\mathcal{F}=(T_k)_{k}$ be a set of subspaces of $\mathbb{C}^{n}$ with $\sum_{k}\text{dim}(T_{k}) \leq d$, let $M_k$ be the orthonormal projection onto $T_k$ and let $Q= \left\{\rho  : \rho \text{ is a density matrix on }\mathbb{C}^{n} \text{with Tr} (\rho M_k) >\delta  \text { for at least } m \text{ many } k\right\}.$ Let $V=\sum_k M_k$. If $C$ is a maximal (non-extendable) set of orthonormal vectors $u$ with the property that $\big<u|Vu\big> > m\delta/4$ and if $M$ is the orthonormal projection onto span$(C)$, then
we have that \[\text{Tr}(M) < \dfrac{4d}{\delta m}   \text{ and Tr}(M \rho) > \dfrac{\delta}{4}  \text{ for every } \rho  \in Q.\]
 
\end{thm}
\begin{proof}
We need the following notions; for an operator $A$, $||A||:=$max$_{u\neq 0}\big<u|Au\big>/\big<u|u\big>$. That $||A||$ equals the maximum eigenvalue of $A$ and that $||A+B||\leq ||A||+||B||$ for operators $A$ and $B$ are standard results. The proof of Theorem \ref{thm:LinAStr} relies on the following lemma.
\begin{lem}
\label{review}
Let $V$ be a positive Hermitian operator on a $n$-dimensional Hilbert space and let $m,\delta >0$. Let $C$ be a maximal (non-extendable) set of orthonormal vectors $u$ with the property that $\big<u|Vu\big> > m\delta/4$. Let $M$ be the orthonormal projection onto span$(C)$. Then, Tr$(M)< 4$Tr$(V)/m\delta$ and for any density matrix $\rho$, if there is a positive Hermitian operator $W \leq V$ with $||W||\leq m$ and Tr$(W\rho)>m\delta$, then Tr$(M\rho)>\delta/4$.
\end{lem}
\begin{proof}
Let $\lambda = m\delta/4$ and let $C=\{w_1,\cdots , w_k \}$ be as in the statement of the lemma. So, Tr$(M)=k$ and we see that Tr$(V) \geq \sum_{i\leq k} \big<w_i | V w_i \big> > k\lambda =$ Tr$(M)\lambda$. I.e., Tr$(M)< 4$Tr$(V)/m\delta $. This gives the desired upper bound on Tr$(M)$. Now, we show the other part of the lemma.

Let $\rho$ and $W$ be as in the statement of the lemma. Let $e_1,..,e_n$ be the orthonormal eigenbasis of $\rho$ with corresponding eigenvalues, $\alpha_1,..,\alpha_n$ which satisfy $\sum_i \alpha_i = 1$.
\begin{align}
\label{eq:str1}
    m\delta < \text{Tr}(W\rho)=\sum_i \big<e_i |W\rho e_i\big>=\sum_i \alpha_i \big<e_i |W e_i\big>.
\end{align}
Similarly,
\begin{align}
\label{eq:str2}
    \text{Tr}(M\rho)=\sum_i \alpha_i \big<e_i |M e_i\big>=\sum_i \alpha_i \big<M e_i |M e_i\big>.
\end{align}
Letting $Me_i = u_i = u, e_i - u_i = v_i=v$,
\[\big<u+v|W|u+v\big>= \big<u|Wu\big>+\big<v|W v\big>+ \big<u|W v\big>+\big<v|Wu\big>.\]
Also, $\big<u|W v\big>+\big<v|Wu\big> \leq \big<u|Wu\big>+\big<v|W v\big>,$
since
\[0\leq \big<u-v|W|u-v\big>= \big<u|Wu\big>+\big<v|W v\big>- \big<u|W v\big>-\big<v|Wu\big>.\]
So, $\big<u+v|W|u+v\big> \leq  2 \big<u|Wu\big>+2\big<v|W v\big>.$ Putting this into\ref{eq:str1} yields;
\begin{align}
\label{eq:str3}
    m\delta< 2\sum_i \alpha_i \big<u_i | W u_i\big>+2\sum_i \alpha_i \big<v_i | W v_i\big>.
\end{align}
Recall that $C$ is a maximal set of orthonormal vectors $u$ with the property that $\big<u|Vu\big> > \lambda$ and that each $v_i$ is orthogonal to every vector in $C$. So, $\big<v_i|V v_i\big> \leq \lambda$ for all $i$. Noting that $\big<v_i | W v_i\big> \leq \big<v_i | V v_i\big> $, we see that the second term in\ref{eq:str3} is $\leq 2 \lambda$. As $||W|| \leq m$, we have that $\big<u_i|Wu_i\big>\leq m\big<u_i|u_i\big>$ for all $i$. Together with\ref{eq:str2} this implies that the first term in\ref{eq:str3} is $\leq 2m\sum_i \alpha_i \big<u_i|u_i\big>=2m\sum_i \alpha_i \big<Me_i|Me_i\big>=2m$Tr$(M\rho)$. So,\ref{eq:str3} becomes: $m\delta< 2m$Tr$(M\rho)+ m\delta/2$. I.e., $\delta/4<$Tr$(M\rho)$ as desired.
\end{proof}
We now prove the theorem. Note that $V=\sum_k M_k$ is positive and Hermitian since each $M_k$ is. Also, Tr$(V)=\sum_k$Tr$(M_k)=\sum_k$dim$(T_k)\leq d$. Let $C$ and $M$ be as in the statement of the theorem. By Lemma \ref{review}, Tr$(M) < 4d/m\delta$, proving the bound. Now, fix some $\rho \in Q$ and let $H$ be a set of $m$ many indices, $k$ such that Tr$(\rho M_k)> \delta$. Let $W:=\sum_{k\in H} M_k$. Then, $W\leq V$ and $||W||\leq \sum_{k\in H} ||M_k||= m$ as $M_k$ is an orthonormal projection and so $||M_k||=1$. Also, Tr$(W\rho)=\sum_{k\in H}$Tr$(M_k\rho)> m\delta$. So, Tr$(M\rho)>\delta/4$ by Lemma \ref{review}. 
\end{proof}

\subsection{Quantum Solovay randomness is equivalent to quantum Martin-L{\"o}f randomness }
\begin{thm}
\label{thm:000}
A state is quantum Solovay random if and only if it is quantum Martin-L{\"o}f random. 
\end{thm}
\begin{proof}
It suffices to show that if a state $\rho$ is not quantum Solovay random then it is not quantum Martin-L{\"o}f random.
To this end, let $\rho=(\rho_n)_{n\in \omega}$ be a state which fails a quantum Solovay test, $(S^{k})_{k\in \omega}$ at level $\delta$. We show that $\rho$ is not quantum Martin-L{\"o}f random by building a quantum Martin-L{\"o}f test, $(G^{m})_{m\in\omega}$, with $G^{m}= (G^{m}_n)_{n\in\omega}$,  which $\rho$ fails at level $ \delta/4$. Without loss of generality, assume that $S^{k}_n =\emptyset$ for $k>n$ and let $\sum_{k} \tau(S^k )<1$. 
 We use the notation:
\[A^{m}_{t}= \left\{\psi \in \mathbb{C}^{2^t}: ||\psi||=1, \sum_{k\leq t} \text{Tr}(|\psi \big> \big< \psi| S^{k}_{t})> \dfrac{2^{m}\delta}{4}\right\}.\]
\emph{Construction of $G^{m}$}: 
We build $G^m$ inductively as follows. Given $C^{m}_{n-1} \subseteq \mathbb{C}^{2^{n-1}}_{alg}$, a maximal (non-extendable) orthonormal subset of $A^{m}_{n-1}$, let \[D^{m}_{n}=\left\{|\psi \big>\otimes |i\big> \in \mathbb{C}^{2^n}_{alg}: i \in \{1,0\},     \psi \in C^{m}_{n-1}\right\}.\]
Note that $D^{m}_n \subseteq A^{m}_n$ since $C^{m}_{n-1} \subseteq A^{m}_{n-1}$. 
Compute $S$, a maximal orthonormal subset with the property that $D^{m}_n \subseteq S \subseteq A^{m}_n $ and define $C^{m}_n := S$. $C^{m}_n $ is composed of complex algebraic vectors by construction since $D^{m}_n \subseteq C^{2^n}_{alg}$ and as we may assume that the operations required to compute $S$ from $D^{m}_n$ are algebraic (See Remark \ref{S}). Let $G^{m}_{n}$ be the projection: \[G^{m}_{n}= \sum_{\psi \in C^{m}_{n}}|\psi\big> \big< \psi|.\] \emph{End of construction}.

\begin{lem}
$(G^{m})_{m\in\omega}$ is a quantum Martin-L{\"o}f test.
\end{lem}

\begin{proof}
Fix $m$. Clearly,  $(C^{m}_{n})_{n\in\omega}$ is a uniformly computable sequence. By construction, $\text{range}(G^{m}_{n-1} \otimes I_{2}) \subseteq \text{range}(G^{m}_n)$.
So, $G^{m}=(G^{m}_{n})_{n\in\omega}$ is a  quantum-$\Sigma^{0}_1$ set for each $m$.
The sequence $(G^{m})_{m\in\omega}$ is uniformly computable in $m$ by construction. Since, $ 1 \geq \sum_{k} \tau(S^{k})$, we have that $ 2^{n} \geq \sum_{k} \text{Tr}(S^{k}_{n})$ for all $n$.  Make the replacements: $d \mapsto 2^n , n\mapsto 2^n, m\mapsto 2^m , M_k \mapsto S^k_n$ (and hence $V\mapsto \sum_{k\leq n} S^k_n$) and $C \mapsto C^m_n$ (and hence $M \mapsto G^m_n$) in Theorem \ref{thm:LinAStr} to see that Tr$(G^{m}_{n})  < (4/\delta)  2^{n-m}$ for all $m,n$. So, $\tau(G^{m})  < (4/\delta)  2^{-m}$ for all $m$.
\end{proof}
\begin{lem}
\label{lem:useweakcor}
$\rho$ fails $(G^m)_m$ at level $\dfrac{\delta}{4}$.
\end{lem}
\begin{proof}
We must show that $\text{inf}_{m\in\omega} \rho(G^m) > \dfrac{\delta}{4}.$ It suffices to show that for all $m\in\omega$, there is an $n$ such that 
$\text{Tr}(\rho_n G^{m}_n)> \dfrac{\delta}{4}.$ To this end, let $m$ be arbitrary and fix a $n$ big enough so that there exist $2^{m}$ many $ks$ such that Tr$(\rho_n S^{k}_n) > \delta$. Make the replacements in Theorem \ref{thm:LinAStr}: $d \mapsto 2^n , n\mapsto 2^n, m\mapsto 2^m , M_k \mapsto S^k_n$ (and hence $V\mapsto \sum_{k\leq n} S^k_n$) and $C \mapsto C^m_n$ (and hence $M \mapsto G^m_n$). As $C^m_n$ is a maximal orthonormal set in $A^m_n$, Theorem \ref{thm:LinAStr} gives that Tr$(\rho_n G^m_n)>\delta/4$.
\end{proof}
Theorem \ref{thm:000} is proved.\qedhere
\end{proof}
\begin{remark}
\label{S}
 We explain a construction method which ensures that $S$ is complex algebraic. $S$ is computed by initially setting $X:=D^m_{n}$ and then by progressively extending $X$ in stages, maintaining  orthonormality and that $X\subseteq A^m_{n}$, until it is no longer possible to do so, at which stage we let $S:=X$. The details are as follows. Initialize $X$ to $D^m_{n}$. At each stage, we do the following: Let $P$ be the orthonormal projection onto span$(X)^{\perp}$. Note that if $\theta$ is the maximum eigenvalue of $P \big(\sum_{k\leq n}S^k_n\big ) P$, and if $w$ is the unit length eigenvector with eigenvalue $\theta$, then $w\in$ span$(X)^{\perp}$,
 \[\max\limits_{\psi \in \text{span}(X)^{\perp}, ||\psi||=1}\big<\psi|\sum_{k\leq n}S^k_n|\psi\big>=\max\limits_{\psi \in \text{span}(X)^{\perp}, ||\psi||=1}\big<\psi|P\big( \sum_{k\leq n}S^k_n \big) P|\psi\big>= \theta,\]
 and
 \[\theta=\big<w|P \big(\sum_{k\leq n}S^k_n\big ) P|w\big>=\big<w| \sum_{k\leq n}S^k_n |w\big>.\]

 So we extend $X$, or check that it is the desired $S$, as follows: Compute $\theta$, the maximum eigenvalue of $P \big(\sum_{k\leq n}S^k_n\big ) P$, and $w$, the unit length eigenvector corresponding to $\theta$. If $\theta >2^m\delta/4$, then we add $w$ to $X$ hence extending it and go to the next stage. If $\theta \leq 2^m\delta/4$, then $X$ is maximal in $A^m_{n}$ and so we set $S=X$. It is easy to see that the operations of finding $w$ and $\theta$ at each stage are algebraic and hence that $w$ is complex algebraic.
\end{remark}

\subsection{Convexity}
We show that all classes of random states are convex. The first result in this section is a corollary of the main theorem from the previous section.
\begin{cor}

A convex combination of q-Martin-L{\"o}f random states is q-Martin-L{\"o}f random. Formally, if $(\rho^i )_{i<k<\omega}$ are q-ML random states and $\sum_{i<k<\omega} \alpha_{i} =1$, then $\rho=\sum_{i<k}\alpha_{i}\rho^{i}$ is q-ML random.
\end{cor}
\begin{proof}
Suppose for a contradiction that there is a q-Martin-L{\"o}f test $(G^{m})_{m\in \omega}$  and a $\delta>0$ such that $\forall m\in \omega$, $ \rho(G^m)>\delta$.  So, $\forall m \in \omega$, $\exists n$ such that Tr$(\rho_{n}G^{m}_{n})>\delta $ where $\rho_n =\sum_{i<k}\alpha_{i}\rho^{i}_n$. So, $\forall m \in \omega$, $\exists n$ such that \[\delta<\text{Tr}\bigg(\sum_{i<k}\alpha_{i}\rho^{i}_{n}G^{m}_{n}\bigg)= \sum_{i<k}\alpha_{i}\text{Tr}(\rho^{i}_{n}G^{m}_{n}).\]
By convexity of the sum, there is an $i$ such that Tr$(G^{m}_{n}\rho^{i}_{n})>\delta$. In summary, \[\forall m, \text{ there is an } i  \text{ and an }n  \text{ such that Tr} (\rho^{i}_{n}G^{m}_{n})>\delta.\]
Since there are only finitely many $i$s, by the pigeonhole principle, there is an $i$ such that $\exists^{\infty} m$ with Tr $ (\rho^{i}_{n}G^{m}_{n})>\delta$, for some $n$. So, $\exists^{\infty} m$ with $ \rho^{i} (G^{m})>\delta$. So, $\rho^i$ fails the q-Solovay test $(G^{m})_{m\in \omega}$ and hence is not q-Martin-L{\"o}f random by our previous result. This is a contradiction. 
\end{proof}
\begin{thm}

A convex combination of quantum Schnorr random states is quantum Schnorr random. Formally, if $(\rho^i )_{i<k<\omega}$ are quantum Schnorr random states and $\sum_{i<k<\omega} \alpha_{i} =1$, then $\rho=\sum_{i<k}\alpha_{i}\rho^{i}$ is quantum Schnorr random.
\end{thm}
\begin{proof}
Suppose for a contradiction that there is a quantum Schnorr test $(G^{m})_{m\in \omega}$  and a $\delta>0$ such that $\exists^{\infty} m\in \omega$, $ \rho(G^m)>\delta$. Letting $G^m$ be $n_{m}$ by $n_{m}$, $\exists^{\infty} m$, such that \[\delta<\text{Tr}(\rho_{n_{m}}G^{m}_{n_{m}})=\text{Tr}(\sum_{i<k}\alpha_{i}\rho^{i}_{n_{m}}G^{m}_{n_{m}})= \sum_{i<k}\alpha_{i}\text{Tr}(\rho^{i}_{n_{m}}G^{m}_{n_{m}}).\]
By convexity of the sum, there is an $i$ such that Tr$(G^{m}_{n_{m}}\rho^{i}_{n_{m}})>\delta$. In summary, \[\exists^{\infty} m, \text{ there is an } i   \text{ such that Tr} (\rho^{i}_{n_{m}}G^{m}_{n_{m}})>\delta.\]
Since there are only finitely many $i$ s, by the pigeonhole principle, there is an $i$ such that $\exists^{\infty} m$ with Tr $ (\rho^{i}_{n_{m}}G^{m}_{n_{m}})>\delta$. So, $\exists^{\infty} m$ with $ \rho^{i} (G^{m})>\delta$. So, $\rho^i$ fails the q-Schnorr test $(G^{m})_{m\in \omega}$ and hence is not q-Schnorr. This is a contradiction.
\end{proof}
Noting that the above proof needed only the Solovay type of failing criterion, we get:
\begin{thm}

A convex combination of weak Solovay random states is q-weak Solovay  random. Formally, if $(\rho^i )_{i<k<\omega}$ are weak Solovay random states and $\sum_{i<k<\omega} \alpha_{i} =1$, then $\rho=\sum_{i<k}\alpha_{i}\rho^{i}$ is weak Solovay  random.
\end{thm}
The proof is almost identical to the previous one.
\subsection{Nesting property of quantum Martin-L{\"o}f tests}
It is interesting to see which classical results carry over to the quantum realm. For example, the existence of a universal MLT, $(U_{n})_n$ such that a bitstring is MLR if and only if it passes this $(U_{n})_n$ does carry over \cite{unpublished}. The `nesting property' of the classical Martin-L{\"o}f test says that we can, without loss of generality assume the universal test $(U_{n})_n$ to be nested; i.e., to satisfy $U_{n+1} \supseteq U_n$ for all $n$. We extend this property to the quantum setting:
\begin{thm}
\label{thm:00}
There is a q-MLT, $(Q^{m})_{m\in \mathbb{N}}$ with the properties (1) If a state $\rho$ fails the universal q-Martin-L{\"o}f test $(G^{m})_{m\in \mathbb{N}}$ at $\delta>0$, then, it also fails $(Q^{m})_{m\in \mathbb{N}}$ at $\delta>0$ (2) If $Q^{m}=(Q^{m}_{n})_{n\in \mathbb{N}}$ for all m, then for all $m$ and $n$, range$(Q^{m+1}_{n}) \subseteq $ range $(Q^{m}_{n}) $. In particular, $Q^{m+1} \leq Q^m$ for all m.
\end{thm}
\begin{proof}
Informally speaking, we want to let $Q^m$ be $\sum_{i>m} G^{i}$. Precisely, we build $Q^m$ level by level. For any natural numbers $i \leq n$, let\[G^{i}_{n} = \sum_{j=1}^{2^{n-i}} |v^{i,n}_{j}\big>\big<v^{i,n}_{j}|.\] Let 
\[S^{m}_{n}:= \text{span} \bigcup_{i=m}^{n} \{v^{i,n}_{j}: 1\leq j \leq 2^{n-i}\},
\]
and let $Q^{m}_{n}$ be the special projection onto $S^{m}_{n}$. Let $Q^{m}= (Q^{m}_{n})_{n}$. Fix an $m$. We see that $Q^{m}_{n} \leq Q^{m}_{n+1}$, since $G^{i}_{n} \leq G^{i}_{n+1}$ holds for all $i,n$. So, $Q^m$ is a q-$\Sigma^{0}_1$ class. The dimension of $S^{m}_{n}$ is at most $\sum_{i=m}^{n} \text{Trace}(G^{i}_{n}) \leq \sum_{i=m}^{n} 2^{n-i} < 2^{n-m+1}.$ So, $(Q^{m})_{m=2}^{\infty}$ is a q-MLT. Let $m$ and $n$ be arbitrary and $n\geq m+1$. Then, clearly, by definition of $S^{m}_{n}$, we see that range$(Q^{m+1}_{n}) \subseteq $ range $(Q^{m}_{n}) $. So, the nesting property holds. Let $\rho= (\rho_{n})_{n}$ be a state. By the nesting, and by properties of projection operators, we have that for a fixed $m$ and all $n$,
\[\text{Tr}(\rho_{n}Q^{m+1}_{n}) \leq \text{Tr}(\rho_{n}Q^{m}_{n}) \leq \rho(Q^{m}).\]
So, $\rho(Q^{m+1})=$ sup$_{n}\text{Tr}(\rho_{n}Q^{m+1}_{n}) \leq \rho(Q^{m})$ for all $m$. (1) clearly holds.
\end{proof}

\section{Randomness for diagonal states}
\label{sec:measures}

A state $\rho=(\rho_{n})_n$ is defined to be \emph{diagonal} if $\rho_n$ is diagonal for all $n$. So, each $\rho_n$ in a diagonal state represents a  mixture of separable states. A diagonal $\rho=(\rho_{n})_n$ can be thought of as a measure on Cantor space, denoted by $\mu_{\rho}$: if $\sigma \in 2^n$, we define $\mu_{\rho}(\llbracket\sigma\rrbracket):= \big<\sigma|\rho_{n}|\sigma\big>$. We will write $\mu_{\rho}(\sigma)$ instead of $\mu_{\rho}(\llbracket\sigma\rrbracket)$. $\mu_{\rho}$ is easily seen to be a measure by noting that the partial trace over the last qubit of $\rho_{n+1}$ equals $\rho_n$ for all $n$. Recalling the notation in Remark \ref{rem:notation} and as $S$ is prefix free, we have, \[\mu_\rho (\llbracket S \rrbracket)= \sum_{\sigma \in S} \mu_\rho (\sigma) =\sum_{\sigma \in S} \big<\sigma|\rho_{n}|\sigma\big>=\text{Tr}(\rho_n P_S).\]
This will be used frequently. Nies and Stephan have recently defined a notion of randomness for measures on Cantor space called Martin-L{\"o}f absolute  continuity\cite{unpublished2}.
\begin{defn}
A measure $\pi$ on Cantor space is called \emph{ Martin-L{\"o}f absolutely continuous} if $\inf_{m}\pi(G_{m})=0$ for each classical MLT $(G_{m})_{m\in \mathbb{N}}$.
\end{defn}
This notion turns out to be equivalent to quantum Martin-L{\"o}f randomness in the sense that for a diagonal $\rho$, $\rho$ is q-MLR if and only if $\mu_\rho$ is  Martin-L{\"o}f absolutely continuous. It is easy to see that if a diagonal $\rho$ is q-MLR, then $\mu_\rho$ is Martin-L{\"o}f absolutely continuous. We now show the other direction.
\begin{thm}
\label{thm:30}
Let $\rho$ be diagonal. If it fails a q-MLT $(G^{m})_{m\in \mathbb{N}}$ at order $\delta$, then there is a classical MLT, $(C^{m})_{m\in \mathbb{N}}$ such that  $\inf_{m}\mu_{\rho}(C^{m})>\delta/2$.
\end{thm}
\begin{proof}
We isolate here a simple but useful property. 
\begin{lem}
\label{lem:30}
Let $n$ be a natural number, $E=(e_{i})_{i=1}^{2^{n}}$ be any orthonormal basis for $\mathbb{C}^{2^{n}}$ and $F$ be any Hermitian, orthonormal projection matrix acting on $\mathbb{C}^{2^{n}}$. For any $\delta>0$, let  \[S^{\delta}_{E,F}:= \left\{e_{i} \in E: \big<e_{i}|F|e_{i}\big> > \delta \right\}.\]
Then, $|S^{\delta}_{E,F}| < \delta^{-1} \text{Tr}(F)$.
\end{lem}
\begin{proof}
Note that since $F$ is a Hermitian orthonormal projection, $\big<e_{i}|F|e_{i}\big>=\big<Fe_{i}|F e_{i}\big> = |F e_{i}|^{2}\geq 0$. So,
\[\delta|S^{\delta}_{E,F}| < \sum_{e_{i} \in S^{\delta}_{E,F}}\big<e_{i}|F|e_{i}\big> \leq \sum_{i \leq 2^{n}}\big<e_{i}|F|e_{i}\big> = \text{Tr}(F).\]
\qedhere
\end{proof}
We now prove Theorem \ref{thm:30}. The intuition is as follows: given a special projection, we take the set of bitstrings (thought of as qubitstrings) `close' to it. If the special projection `captures' $\delta$ much mass of $\rho$, then the projection onto the span of these qubitstrings must capture atleast $\delta/2$ much mass of $\rho$.
$\sigma$ will always denote a finite length classical bit string and $|\sigma\big>$, the corresponding element of the standard computational basis. We may assume that $\delta$ is rational. Fix $m$. We describe the construction of $C^{m}= (C^{m}_{n})_{n\in \mathbb{N}}$ (See \ref{def:10}). Let  \[T^{m}_{n}:= \left\{\sigma \in 2^{n}: \big<\sigma|G^{m}_{n}|\sigma\big> > \dfrac{\delta}{4}\right\}.\]
These are those standard basis vectors `close' to $G^{m}_{n}$. Let
\[C^{m}_{n}= \bigcup_{\sigma \in T^{m}_{n}} \llbracket \sigma \rrbracket.\] 
\begin{lem}
$C^{m} $ is a $\Sigma^{0}_1$ class for any $m$.
\end{lem}
\begin{proof}
It is easy to see that for all $\sigma \in T^{m}_{n}$ and $i\in \{0,1\}$, \[\big<\sigma i|G^{m}_{n+1}|\sigma i \big > \geq \big<\sigma|G^{m}_{n}|\sigma\big> > \dfrac{\delta}{4}.\]
So, $\{\sigma i: \sigma \in T^{m}_n , i\in \{0,1\}\} \subseteq T^{m}_{n+1}$. Also note that $T^{m}_{n}$ is uniformly computable in $n$ since $G^{m}_{n}$ is.
\end{proof}
\begin{lem}
$(C^{m})_{m\in \mathbb{N}}$ is a MLT.
\end{lem}

\begin{proof}
Fix $m$. Letting $E=B^n$ and $F=G^{m}_n$ in Lemma \ref{lem:30} and by definition of q-MLT,
\[ |T^{m}_{n}| < \dfrac{4}{\delta}
2^n \tau(G^{m})\leq \dfrac{4}{\delta} 2^{n-m}.\]
So, $\mu(C^{m}) < 2^{-m}\dfrac{4}{\delta}$. $C^{m}$ is computable in $m$ since $G^{m}$ is.
\end{proof}
 
Now we show that $\inf_{m}\mu_{\rho}(C^{m})>\delta/2$.
Fix a $m$ and a $n$ (depending on $m$) such that Tr$(\rho_{n} G^{m}_{n})>\delta$. Let $\rho_{n} = \sum_{\sigma \in 2^{n} } \alpha_{\sigma} |\sigma\big>\big<\sigma|.$
Then, \begin{align*}
&\delta<\text{Tr}(\rho_{n} G^{m}_{n})
=\sum_{\sigma \in 2^{n} } \alpha_{\sigma} \big<\sigma|G^{m}_{n}|\sigma\big> = \sum_{\sigma \in T^{m}_{n}} \alpha_{\sigma} \big<\sigma|G^{m}_{n}|\sigma\big>  + \sum_{\sigma \in 2^{n}\backslash T^{m}_{n} } \alpha_{\sigma} \big<\sigma|G^{m}_{n}|\sigma\big>\\
&\leq \sum_{\sigma \in T^{m}_{n}}  \alpha_{\sigma}  + \sum_{\sigma \in 2^{n}\backslash T^{m}_{n}} \alpha_{\sigma} \dfrac{\delta}{4}\leq \sum_{\sigma \in T^{m}_{n}}  \alpha_{\sigma}  +  \dfrac{\delta}{4}=\text{Tr}(\rho_n  P_{C^{m}_{n}}) + \dfrac{\delta}{4}= \mu_{\rho}(C^{m}_{n}) + \dfrac{\delta}{4}.
\end{align*}
The last equality follows as $T_{n}^m$ is prefix free. So, $\mu_{\rho}(C^{m})\geq \mu_{\rho}(C^{m}_{n})\geq 3\delta/4$.
%
\end{proof}

Nies and Scholz showed that a measure, $\mu$ is Martin-L{\"o}f absolutely continuous if and only if for any Solovay test $(S_k)_k$, $\lim_{k}\mu (S_k )=0$\cite{unpublished2}. We give another proof using ideas used in the proof of Theorem \ref{thm:000}.
\begin{thm}
\label{thm:2}
Let $\rho$ be diagonal. If for some Solovay test $(S_k)_k$ and $\delta>0$ we have $\exists^{\infty} k , \mu_\rho (S_k )>\delta$, then there is a Martin-L{\"o}f test $(J_m )_m$ such that $\inf_{m}\mu_\rho (J^m )>\delta/2$. 
\end{thm}
The theorem will follow from the two lemmas below. Write $S^k = (S^{k}_n )_n$ as in Definition \ref{def:10}.
Without loss of generality, let $S^{k}_n =\emptyset$ for $k>n$. Let
\[C^{m}_{t}= \left\{ \sigma  \in 2^t : \sum_{k\leq t} |\big< \sigma| S^{k}_{t}|\sigma \big>|> 2^{m-1}\delta\right\},\] 
and let $G^{m}_{t}:= P_{C^{m}_{t}}$ (See Remark \ref{rem:notation}).
Let $G^{m}=(G^{m}_{n})_{n}$. It is easy to see that $G^m$ is a q-$\Sigma^{0}_1$ set for each $m$. Let $J^{m}_n := \llbracket C^{m}_n \rrbracket$ and $J^m = (J^{m}_n)_n$. One can check that that $(J^m )_m$ is a MLT if and only if $(G^{m})_{m}$ is quantum Martin-L{\"o}f test. So, $(J^m )_m$ is a MLT since:

\begin{lem}
$(G^{m})_{m}$ is a quantum Martin-L{\"o}f test.
\end{lem}

\begin{proof}
Note that we may assume that $\sum_k \tau(S^k)<1$. Fix $m$ and $t$, let $C:= C^m_{t}$.
\begin{align*}
     2^t &\geq \sum_{k\leq t}  \text{Tr}(S^k_{t})\geq \sum _{k\leq t}  \sum_{\sigma \in C}   \big< \sigma|S^k_{t}|\sigma\big> =\sum_{\sigma\in C}\sum_{k\leq t}   \big< \sigma|S^k_{t}|\sigma\big>\\
            &>|C|2^{m-1}\delta =\text{Tr}(G^m_{t})2^{m-1}\delta.\qedhere
 \end{align*}
 So, $\tau(G^m)<2^{-m+1}/\delta$ for all $m$.
\end{proof}
\begin{lem}
We have that $\inf_{m}\mu_\rho (J^m )>\delta/2$.
\end{lem}
\begin{proof}
Let $m$ be arbitrary. By assumption, there are infinitely many $k$s such that  $\mu_\rho (S^{k})>\delta$. For each of these, there is an $n$ such that $\mu_\rho (S^{k}_n )>\delta$. So, fix a $n$ so that there are $2^{m}$ many $ks$ such that $\mu_\rho( S^{k}_n) > \delta$.
Since $\rho_n$ is diagonal, let 
 \[\rho_n = \sum_{\sigma \in 2^n} \alpha_{\sigma}|\sigma \big> \big< \sigma |.\]  

By the choice of $n$, pick $M\subseteq\{1,2\cdots ,n\}$ such that $|M|=2^m$ and $ \mu_\rho (S^{k}_n) > \delta$ for each $k$ in $M$. Note that $\mu_\rho (S^{k}_n)= \text{Tr}(\rho_n P_{S^{k}_n})$, since $S^{k}_n$ is prefix free. We write $\text{Tr}(\rho_n P_{S^{k}_n})=\text{Tr}(\rho_n S^{k}_n)$ to avoid clutter. So,

\begin{align*}
    2^{m}\delta
    &<\sum_{k\in M}\mu_\rho (S^{k}_n)= \sum_{k\in M}\text{Tr}(\rho_n S^{k}_n)
    =   \sum_{\sigma \in 2^n}\alpha_{\sigma} \sum_{k\in M} \text{Tr}(|\sigma \big> \big< \sigma | S^{k}_n)\\
    &=\sum_{ \sigma \in C^{m}_{n}}\alpha_{\sigma} \sum_{k\in M} \big< \sigma |S ^{k}_n\sigma \big> + \sum_{ \sigma \notin C^{m}_{n}}\alpha_{\sigma} \sum_{k\in M} \big< \sigma |S ^{k}_n\sigma \big>\\
    &\leq \sum_{ \sigma \in C^{m}_{n}}\alpha_{\sigma} \sum_{k\in M} \big< \sigma |S ^{k}_n\sigma \big> + 2^{m-1}\delta\\
    &\leq 2^{m} \sum_{ \sigma \in C^{m}_{n}}\alpha_{\sigma} + 2^{m-1}\delta.
\end{align*}
The second last inequality follows from the definition of $G^{m}_{n}$ and convexity; the last from the choice of $M$.
Finally, we get that,
\[\delta/2 < \sum_{\sigma \in C^{m}_{n}}\alpha_{\sigma} = \mu_\rho(\llbracket C^{m}_{n}\rrbracket)\leq \mu_\rho (J^m ).\] \qedhere
\end{proof}

Next, we discuss a subset of the diagonal states; the Dirac delta measures on Cantor space.
\subsection{Quantum randomness on Cantor Space}
\label{schnorr}
A classical bitstring can be thought of as a diagonal state: If $X$ is a real in Cantor space, the state $\rho_{X}=(\rho_{n})_n$ given by $\rho_{n}=|X\upharpoonright n \big>\big<X\upharpoonright n |$ is the quantum analog of $X$. Do the quantum randomness notions agree with classical notions when restricted to Cantor space? By Theorem \ref{thm:30}, we see that $\rho_{X}$ is q-MLR if and only if $X$ is MLR. Further, $\rho_{X}$ is q-MLR if and only if $\rho_{X}$ is weak Solovay random\cite{prep}. Also, $X$ is MLR if and only if it passes all interval Solovay tests (the classical analog of strong Solovay tests). So, we see that q-MLR and weak Solovay randomness agree with the classical versions on Cantor space. What about quantum Schnorr randomness?
\begin{lem}
\label{thm:schcla}
$\rho_{X}$ is quantum Schnorr random if and only if $X$ is Schnorr random.
\end{lem}
\begin{proof}
Let $(Q^r)_r$ be a quantum Schnorr test which $\rho_X$ fails at some rational $\delta$. Let $Q^r$ be $n_r$ by $n_r$. Using notation of Lemma \ref{lem:30}, let $T^{r}:= S^{\delta}_{E,Q^{r}}$ where $E$ is the set of length $n_r$ standard basis vectors. We think of $T^r$ as a set of classical bitstrings. By Lemma \ref{lem:30}, $\tau(T^{r})\leq \delta^{-1}\tau(Q^{r})$. So,  $\sum_{r}2^{-n_{r}}|T^{r}|=\sum_{r}\tau(T^{r})  \leq \delta^{-1}\sum_{r}\tau(Q^{r})$ is computable because $\sum_{r}\tau(Q^{r})$ is. So, $(T^r)_r$ is a finite total Solovay test. Let $m$ be one of the infinitely many $r$ such that $\delta<$Tr$(\rho_{X} (Q^r))$. Then, by definition, $X\upharpoonright n_r $ is in  $T^r$. So, $X$ fails $(T_{r})_r$ and hence is not Schnorr random (by 7.2.21 and 7.2.22 in the book by Downey and Hirschfeldt\cite{misc1}). The other direction is trivial. \qedhere

\end{proof}
\subsection{Relating the randomness notions}
We have seen that \[\text{Solovay R} = \text{q-MLR} \subseteq \text{weak Solovay R} \subset \text{quantum Schnorr R}.\]
The equality follows by Theorem \ref{thm:000}. The second inclusion is strict as there is a bitstring which is Schnorr random but which fails some interval Solovay test \cite{misc1} and since by Lemma \ref{thm:schcla}, this bitstring must be quantum Schnorr random. It is open whether the first inclusion is strict.
\section{A law of large numbers for quantum Schnorr randoms}
\label{sec:lln}
The law of large numbers (LLN), specialized to Cantor space says that the limiting proportion of ones is equal to 0.5 for almost every bitstring, with respect to the uniform measure. Random bitstrings satisfy the LLN. In fact, satisfying the LLN is the weakest form of randomness \cite{misc1}. This is quite intuitive; one would not call a bitstring `random' if it has more ones than zeroes in the limit. Analogously, we expect even our weakest notion of quantum randomness (quantum Schnorr randomness) to satisfy a quantum analogue of the LLN. This would suggest that the quantum randomness notions are `natural' and mirror the classical situation. In this section, $\sigma$ will always denote a classical bitstring or a  classical bitstring thought of as a qubit string. It should be clear from context which interpretation is intended.

\begin{defn}
\cite{logicblog} $\rho$ \emph{satisfies the LLN} if lim$_{n}n^{-1}\sum_{i<n}$Tr$(\rho_{n}P^{n}_{i})=0.5$, where for all $i\geq 0, n>0$, \[P^{n}_{i}:= \sum_{\sigma: |\sigma|=n, \sigma(i)=1}|\sigma\big>\big<\sigma| = \bigotimes_{k=1}^{i-1} I_2 \otimes \begin{bmatrix}

0 & 0 \\
0 & 1\\

\end{bmatrix} \otimes \bigotimes_{k=i+1}^n I_2,\]

\end{defn}
where $I_2$ is the two-by-two identity matrix. The intuition is that $P^{n}_i$ is the projection observable which measures whether a given density matrix on $n$ qubits `has a one in the $i^{th}$ spot'. Tr$(\rho_{n}P^{n}_{i})$ is the probability that $\rho_n$ `has a one in the $i^{th}$ spot'. A state satisfies the LLN if the average over $i$ of these probabilities tends to $0.5$ as $n$ goes to infinity.
\begin{thm}
\label{thm:lln}
Quantum Schnorr random states satisfy the LLN.
\end{thm}
Theorem \ref{thm:llngen}, a more general version of the law of large numbers implies Theorem \ref{thm:lln}. We need some definitions to state the general version:
\begin{defn}
Fix a computable $p \in [0,1]$. Let $b_p=(b_{n})_{n}$ be a state of the form $b_{n}=\otimes_{1}^{n} M_p$ where
\[M_p =
\begin{bmatrix}

p & 0 \\
0 & 1-p\\

\end{bmatrix}.\]
A computable sequence of special projections is a $p-$\emph{quantum Schnorr test} if $\sum_{k\in \omega} b_p(S^{k})$ is computable. A state $\rho$ is said to be \emph{p-quantum Schnorr random} if it passes all p-quantum Schnorr tests.
\end{defn} Quantum Schnorr random states are $\dfrac{1}{2}-$ quantum Schnorr random states as $b_{1/2}=\tau$. $b_p$ is the quantum analogue of the p-Bernoulli measure on Cantor space\cite{misc1}. For any computable $p\in [0,1]$, define p-LLN, a general version of the LLN, as:
\begin{defn} $\rho$ \emph{satisfies the p-LLN} if for any pair $a,b$ of computable reals, lim$_{n}n^{-1}\sum_{i<n}$Tr$(\rho_{n}Q^{n}_{i})=ap+b(1-p)$, where for all $n \geq i> 0$, \[Q^{n}_{i}:=  \bigotimes_{k=1}^{i-1} I_2 \otimes \begin{bmatrix}

a & 0 \\
0 & b\\

\end{bmatrix} \otimes \bigotimes_{k=i+1}^n I_2.\]

\end{defn}
The computational basis is an eigenbasis for $Q^{n}_i$ where the vectors with $0$ in the $i^{th}$ qubit have eigenvalue $a$ and those with $1$ in the $i^{th}$ qubit have eigenvalue $b$. So, Tr$(\rho_{n}Q^{n}_{i})$ is the expected value of the random variable which equals $a$ if measuring the $i^{th}$ qubit of $\rho_n$ yields zero and equals $b$ if the measurement yields one. If the average over $i$ of these expected values tends to $pa+(1-p)b$ as $n$ goes to infinity, then the state satisfies the p-LLN. Almost every bitstring, with respect to the p-Bernoulli product measure on Cantor space, satisfies the classical analogue of the p-LLN. The following is hence a classical `almost all' theorem which carries over to quantum Schnorr randoms.
\begin{thm}
\label{thm:llngen}
p-quantum Schnorr random states satisfy the p-LLN.
\end{thm}
\begin{proof}
 Let $\rho$ be a p-quantum Schnorr random. Towards a contradiction, let $a,b$ and $M:=ap+(1-p)b$ be reals so that for $Q^n_i$ as above, there is a rational $\delta>0$ such that either $\exists^{\infty}n$, with $n^{-1}\sum_{i \leq n}$Tr$(\rho_{n}Q^{n}_{i})> \delta + M$ or  $\exists^{\infty}n$, with $n^{-1}\sum_{i\leq  n}$Tr$(\rho_{n}Q^{n}_{i})< -\delta + M$. Suppose first that the former holds and that $M<1$.
Define for all $n$,
\[C_{n}= \left\{\sigma : |\sigma|=n, n^{-1}\sum_{i\leq n} |\big< \sigma| Q^{n}_{i}|\sigma \big>|> (\delta  + 1) M\right\}.\] 
In other words,
\[C_{n}= \left\{\sigma : |\sigma|=n, n^{-1}[b|\{i \leq n: \sigma(i)=1\}| +a|\{i\leq n: \sigma(i)=0\} |]> (\delta  + 1)M\right\}.\] 
Suppose that $2^n$ is equipped with the product measure induced by letting $\mathbb{P}(0)=p, \mathbb{P}(1)=1-p$ (the product p-Bernoulli measure) and that $X_i$ is a random variable on $2^n$ which equals $a$ if $\sigma(i)=0$ and equals $b$ if $\sigma(i)=1$. Then, $\mathbb{E}X_i = M$, $C_n$ is the event
\[C_{n}= \left\{ n^{-1}\sum_{i=1}^n X_i> (\delta  + 1) M\right\},\]

and $\mathbb{P}(C_{n})  \leq \text{exp}(-\dfrac{2\delta^{2}}{(b-a)^2} nM^2)$ for all $n$ by the version of the Chernoff bound in Theorem 6(i) in the notes by Goemans\cite{Chernoff}.
Letting $S_n$ be the special projection, \[S_{n}:= \sum_{\sigma \in C_{n}} |\sigma\big>\big<\sigma|,\]

we see that \[\text{exp}(-\dfrac{2\delta^{2}}{(b-a)^2} nM^2)\geq \mathbb{P}(C_n)=\sum_{\sigma \in C_n} \mathbb{P}(\sigma)= \sum_{\sigma \in C_n}\big<\sigma|b_n|\sigma\big>=  b_p(S_n).\] So, $\sum_{n} b_p(S_{n})$ is a computable real number. $(S_n)_n$ is a computable sequence as $M$ and $\delta$ are computable. So, $(S_n)_n$ is a p-quantum Schnorr test. 
We need some intermediate results:
\begin{lem}
\label{lem:cor}
For a real number $B$ and a random variable $Y \leq B$, if $\mu <\mathbb{E}Y$, then \[\mathbb{P}\{Y \geq \mu\} \geq \dfrac{\mathbb{E}Y-\mu}{B-\mu}.\]
\end{lem}
\begin{proof}
As $B-Y\geq 0$ and $\mu < \mathbb{E}Y\leq B$, Markov's inequality\cite{Chernoff} gives that,
\[\mathbb{P}\{Y < \mu\}=\mathbb{P}\{B-Y > B-\mu\} \leq \dfrac{B-\mathbb{E}Y }{B-\mu}.\]
So,
\[\mathbb{P}\{Y \geq \mu\} \geq 1- \dfrac{B-\mathbb{E}Y }{B-\mu} = \dfrac{ \mathbb{E}Y - \mu}{B-\mu}.\]
\end{proof}
\begin{lem}
\label{lem:cor2}
Let $A$ be a Hermitian operator, $\rho$ a density matrix in a $n$ dimensional Hilbert space and let $e_1,..,e_n$ be an orthonormal eigenbasis for $A$ with $B\geq $max$_i \big<e_i |A| e_i\big>$ and $ \mu < m \leq $ Tr $(\rho A)$. Let $F_{\mu}$ be the orthonormal projector onto the subspace spanned by those $e_i$s with $\big<e_i |A| e_i\big> \geq \mu$. Then, Tr$(\rho F_{\mu}) \geq \dfrac{m-\mu}{B-\mu}$.
\end{lem}
\begin{proof}
Let $Y$ be the discrete random variable with $\mathbb{P}\{Y=\big<e_i |A| e_i\big>\}=\big<e_i |\rho| e_i\big>$. So, $\mathbb{E}Y $ equals the expected value when $A$ is measured on $\rho$ which in turn equals Tr$(A \rho)$. So, $\mathbb{E}Y = $ Tr$(A \rho)$.  Let $\mu<m\leq $ Tr$(A \rho)=\mathbb{E}Y$.  As $\mathbb{P}\{Y \geq \mu\}=$ Tr$(F_{\mu}\rho)$  and $Y\leq B$,  Lemma \ref{lem:cor} gives, Tr$( F_{\mu}\rho) \geq \dfrac{\mathbb{E}Y-\mu}{B-\mu}\geq \dfrac{m-\mu}{B-\mu}$. 
\end{proof}
Now fix one of the infinitely many $n$ such that  $n^{-1}\sum_{i \leq n}$Tr$(\rho_{n}Q^{n}_{i})> \delta + M$. Make the replacements in Lemma \ref{lem:cor2}: let $n \mapsto 2^n$, $\rho \mapsto \rho_n$, $A \mapsto n^{-1}\sum_{i \leq n} Q^{n}_{i} := A_n$, $e_1,..,e_{2^n}\mapsto $ the standard computational basis on $\mathbb{C}^{2^n}$, $B\mapsto z$ := max$\{a,b\}$. $M<1$ implies that
\begin{align}
\label{eq:del}
    &(1+\delta)M<M+\delta < \text{Tr}(A_n \rho). 
\end{align}
So, we can also make the replacements: $\mu \mapsto (1+\delta)M$ and $m \mapsto M+\delta $. With this, $S_n$ is exactly the projection onto the subspace spanned by those computational basis elements which are eigenvectors of $A_n$ with eigenvalues greater than $(1+\delta)M$. So, we also make the replacement $F_{\mu} \mapsto S_n$. By \ref{lem:cor2}, \[\text{Tr} (\rho_n S_n)\geq \dfrac{M+\delta-(1+\delta)M}{z-(1+\delta)M} \geq \dfrac{(1 - M)\delta}{|z|+|(1+\delta)M|}\geq C\delta,\]
where $\dfrac{(1 - M)}{|z|+|3M|}=C>0$ as $M<1$ and as we may assume that $\delta \leq 1$. As this holds for infinitely many $n$, $\rho$ fails the p-quantum Schnorr test. 

The proof of the case where $M \geq 1$ is omitted as it can be proved similarly to the above case after a simple linear transformation. This finishes the proof in the case where $\exists^{\infty}n$, with $n^{-1}\sum_{i\leq  n}$Tr$(\rho_{n}Q^{n}_{i})>  \delta + M$.
The proof of the other case (I.e.,  $\exists^{\infty}n$, with $n^{-1}\sum_{i\leq  n}$Tr$(\rho_{n}Q^{n}_{i})< -\delta + M$) is omitted too as it can be proved similarly to the above proof after an easy linear transformation.
\end{proof}
Taking $a=0,b=1$ and $p=0.5$ in Theorem \ref{thm:llngen} gives Theorem \ref{thm:lln}.
\section{Conclusion and outlook}

In addition to continuing the investigation of quantum Martin-L{\"o}f randomness, we have introduced quantum Solovay and quantum Schnorr randomness. In this section, we mention some of our unpublished results which are related to the present paper.

Effective measure theory (using `effectively null sets') and Kolmogorov complexity theory (using descriptive complexity of initial segments) are, albeit seemingly unrelated, equivalent approaches to study the randomness of bitstrings. This paper and others\cite{unpublished,qpl,prep,logicblog} have generalized the first approach to the quantum realm. In an attempt at generalizing the second approach, we have defined `quantum-K' (denoted by $QK$), a notion of prefix-free quantum Kolmogorov complexity\cite{prep}. Quantum-K is intended to be a quantum version of $K$, the prefix-free Kolmogorov complexity. For any $\epsilon>0$ and density matrix $\tau$, $QK^{\epsilon}(\tau)$ is a measure of the descriptive complexity of $\tau$ based on prefix-free classical Turing machines. Here, $\epsilon$ is an approximation term similar to the $\epsilon$ in $QC^ \epsilon$\cite{Berthiaume:2001:QKC:2942985.2943376}. All notions of quantum Kolmogorov complexity developed so far, with one exception \cite{Vitnyi2001QuantumKC}, use plain classical machines or quantum Turing machines (which are not prefix-free) \cite{Berthiaume:2001:QKC:2942985.2943376,Mller2007QuantumKC}. 

It is worth mentioning some connections of quantum Schnorr and weak Solovay randomness with $QK$. Weak Solovay random states have a characterization\cite{ prep} in terms of the incompressibility of their initial segments in terms of $QK^{\epsilon}$; $\rho$ is weak Solovay random if and only if  the initial segments of $\rho$ have high $QK^{\epsilon}$ in the limit. Precisely, \[\rho  \text{ is weak Solovay random} \iff \forall \epsilon >0, \text{lim}_{n} QK^{\epsilon}(\rho_{n})-n= \infty.\]Recalling that MLR coincides with Solovay randomness in the classical setting, we see that this is analogous to the following classical result attributed to Schnorr by Chaitin\cite{Chaitin:1987:ITR:24912.24913} :
\[X  \text{ is MLR} \iff \text{lim}_{n}  K(X\upharpoonright n)-n= \infty.\]

A prefix-free Turing machine, $C$ is said to be a computable measure machine if the Lebesgue measure of $\llbracket$dom$(C)\rrbracket$ is a computable real number. For a fixed $C$, we define $K_C$ for a finite bitstring $\sigma$ analogously to $K(\sigma)$: $K_C(\sigma):=$ min$\{|\tau|: C(\tau)=\sigma\}$\cite{downey2004schnorr}. Classical Schnorr randomness has a characterization in terms of $K_{C}$ \cite{downey2004schnorr}:
    \[X  \text{ is Schnorr R}\iff\]\[ \forall \text{ computable measure machines } C,  \exists d \forall n,  K_{C}(X\upharpoonright n)>n-d.\]
     We define $QK^{\epsilon}_{C}$, a version of $QK^{\epsilon}$ using the computable measure machine $C$ and show the following\cite{prep}:
    \[\rho \text { is quantum Schnorr R} \iff \]\[\forall \text{ computable measure machines } C, \forall \epsilon \exists d \forall n, QK^{\epsilon}_{C}(\rho_{n})>n-d.\]
We also have another characterization of quantum Schnorr randoms\cite{prep}: $\rho$ is quantum Schnorr random if and only if for all computable measure machines $C$ and all $\epsilon$, $\forall d \forall^{\infty} n$ QK$^{\epsilon}_{C}(\rho_{n})>n+d$. These connections with $QK$ seem to suggest that weak Solovay and quantum Schnorr randomness are `natural' quantum randomness notions. It still remains to find a complexity based characterization of q-MLR states. An interesting question is whether weak Solovay randomness is equivalent to q-MLR, a positive answer to which will yield a $QK$ based characterization of q-MLR. 

In other work\cite{qpl} we have shown that it is possible to `generate' classical randomness from a quantum source which is not quantum random. It hence seems that it is strictly easier to generate classical randomness than it is to generate quantum randomness. In particular, we construct a computable, non q-MLR state which yields an arithmetically random real with probability one when `measured'\cite{qpl}. This is important given the widespread use of random bits in fields such as cryptography and computer science and as arithmetic randomness is considered to be `true' randomness (as against pseudorandomness)\cite{misc1}. It would be interesting to generalize this result to quantum Solovay and Schnorr randoms.

\section{Acknowledgements}
I thank my thesis advisor Joseph S. Miller for his constant encouragement, support and advice and Andr{\'e} Nies for many helpful discussions. I am indebted to an anonymous referee for several insightful comments and for suggesting simplifications to our original proofs of Theorems \ref{thm:LinAStr} and \ref{thm:000} and for generalizing our original proof of Theorem \ref{thm:lln}. An older draft of this paper\cite{bhojraj2020quantum} contains the proofs in their previous form, prior to incorporating the modifications suggested by the referee.

\bibliographystyle{plain}
\bibliography{references.bib}

\end{document}